\documentclass[conference]{IEEEtran}
\IEEEoverridecommandlockouts
\usepackage[noadjust]{cite}
\usepackage{amsmath,amssymb,amsfonts}
\usepackage{algorithmic}
\usepackage{graphicx}
\usepackage{textcomp}
\usepackage{enumitem}
\usepackage{xcolor}
\usepackage{comment}
\usepackage{mathptmx}
\usepackage{amsthm}
\usepackage{changepage}
\usepackage{mathtools}
\usepackage{nicefrac}
\usepackage{balance}
\usepackage[ruled,lined]{algorithm2e}
\definecolor{NYUlight}{HTML}{8900e1} 	

\newtheorem{thm}{Theorem}
\newtheorem{lem}{Lemma}

\theoremstyle{definition}
\newtheorem{defn}{Definition}

\newtheorem{rem}{Remark}
\def\BibTeX{{\rm B\kern-.05em{\sc i\kern-.025em b}\kern-.08em
    T\kern-.1667em\lower.7ex\hbox{E}\kern-.125emX}}

\usepackage{bbm}

\newcommand{\blue}[1]{{\color{blue}#1}}
\newcommand{\indep}{\raisebox{0.05em}{\rotatebox[origin=c]{90}{$\models$}}}
\usepackage{etoolbox}    
\usepackage{dirtytalk}

\usepackage{caption}
\usepackage{subcaption}
\newtoggle{singlecolumn}
\togglefalse{singlecolumn} 

\SetKwComment{Comment}{/* }{ */}
\SetKwInOut{Input}{Input}
\SetKwInOut{Output}{Output}

\SetCommentSty{mycommfont}

\begin{document}

\title{Distribution-Agnostic Database De-Anonymization Under Synchronization Errors\\
\thanks{This work is supported in part by National Science Foundation grants 2148293, 2003182, and 1815821, and NYU WIRELESS Industrial Affiliates.}}

\author{Serhat Bakirtas, Elza Erkip\\
 NYU Tandon School of Engineering\\
Emails: \{serhat.bakirtas, elza\}@nyu.edu }

\maketitle

\begin{abstract}
    There has recently been an increased scientific interest in the de-anonymization of users in anonymized databases containing user-level microdata via multifarious matching strategies utilizing publicly available correlated data. Existing literature has either emphasized practical aspects where underlying data distribution is not required, with limited or no theoretical guarantees, or theoretical aspects with the assumption of complete availability of underlying distributions. In this work, we take a step towards reconciling these two lines of work by providing theoretical guarantees for the de-anonymization of random correlated databases without prior knowledge of data distribution. Motivated by time-indexed microdata, we consider database de-anonymization under both synchronization errors (column repetitions) and obfuscation (noise). By modifying the previously used replica detection algorithm to accommodate for the unknown underlying distribution, proposing a new seeded deletion detection algorithm, and employing statistical and information-theoretic tools, we derive sufficient conditions on the database growth rate for successful matching. Our findings demonstrate that a double-logarithmic seed size relative to row size ensures successful deletion detection. More importantly, we show that the derived sufficient conditions are the same as in the distribution-aware setting, negating any asymptotic loss of performance due to unknown underlying distributions.
\end{abstract}


\section{Introduction}
\label{sec:introduction}
With the accelerating growth of smart devices and applications, there has been a considerable collection of user-level microdata in private companies' and public institutions' possession which is often shared and/or sold. Although this data transfer is performed after removing the explicit user identifiers, a.k.a. \emph{anonymization}, and coarsening of the data through noise, a.k.a. \emph{obfuscation}, there is a growing concern from the scientific community about the privacy implications~\cite{ohm2009broken}. These concerns were further validated by the success of a series of practical attacks on real data by researchers~\cite{naini2015you,datta2012provable,narayanan2008robust,sweeney1997weaving,takbiri2018matching}. In the light of these successful attacks, recently there has been an increasing effort on the information-theoretic and statistical foundations of \emph{database de-anonymization}, a.k.a. \emph{database alignment/matching/recovery}~\cite{cullina,shirani8849392,dai2019database,kunisky2022strong,tamir2022joint,bakirtas2021database,bakirtas2022matching,bakirtas2022seeded,bakirtas2023database,bakirtas2022database}.

In our recent work we have focused on the database de-anonymization problem under synchronization errors. In~\cite{bakirtas2022matching}, we investigated the matching of Markov databases under synchronization errors only, with no subsequent obfuscation/noise. We showed that the synchronization errors could be detected through a histogram-based detection. Furthermore, we found the noiseless matching capacity to be equal to the erasure bound where locations of deletions and replications are known a priori. More relevantly, in~\cite{bakirtas2022seeded}, we considered the de-anonymization of databases under noisy synchronization errors. We proposed a noisy replica detection algorithm and a seeded deletion detection algorithm to recover synchronization errors. We proposed a joint-typicality-based matching algorithm and derived achievability results, which we subsequently showed to be tight, given a seed size logarithmic with the row size of the database. Then in~\cite{bakirtas2023database}, we improved this sufficient seed size to one double logarithmic with the row size. Albeit successful in deriving detecting and matching results, in these works, the availability of information on the underlying distributions was assumed and the proposed algorithms were tailored for these known distributions. 

Motivated by most practical settings where the underlying distributions are not readily available, but only could be estimated from the available data, in this paper, we investigate the de-anonymization problem without any prior knowledge of the underlying distributions. We focus on a noisy random column repetition model borrowed from~\cite{bakirtas2022seeded}, as illustrated in Figure~\ref{fig:intro}. We modify the noisy replica detection algorithm proposed in~\cite{bakirtas2022seeded} so that it still works in the distribution-agnostic setting. Then we propose a novel outlier-detection-based deletion detection algorithm and show that when seeds, whose size grows double logarithmic with the number of users (rows), are available, the underlying deletion pattern could be inferred. Finally, through a typicality-based de-anonymization algorithm that relies on the estimated distributions, we show that database de-anonymization could be performed with no asymptotic loss of performance compared to when all the information on the distributions is available a priori. 

\begin{figure}[t]
\centerline{\includegraphics[width=0.4\textwidth,trim={0cm 9cm 17cm 0},clip]{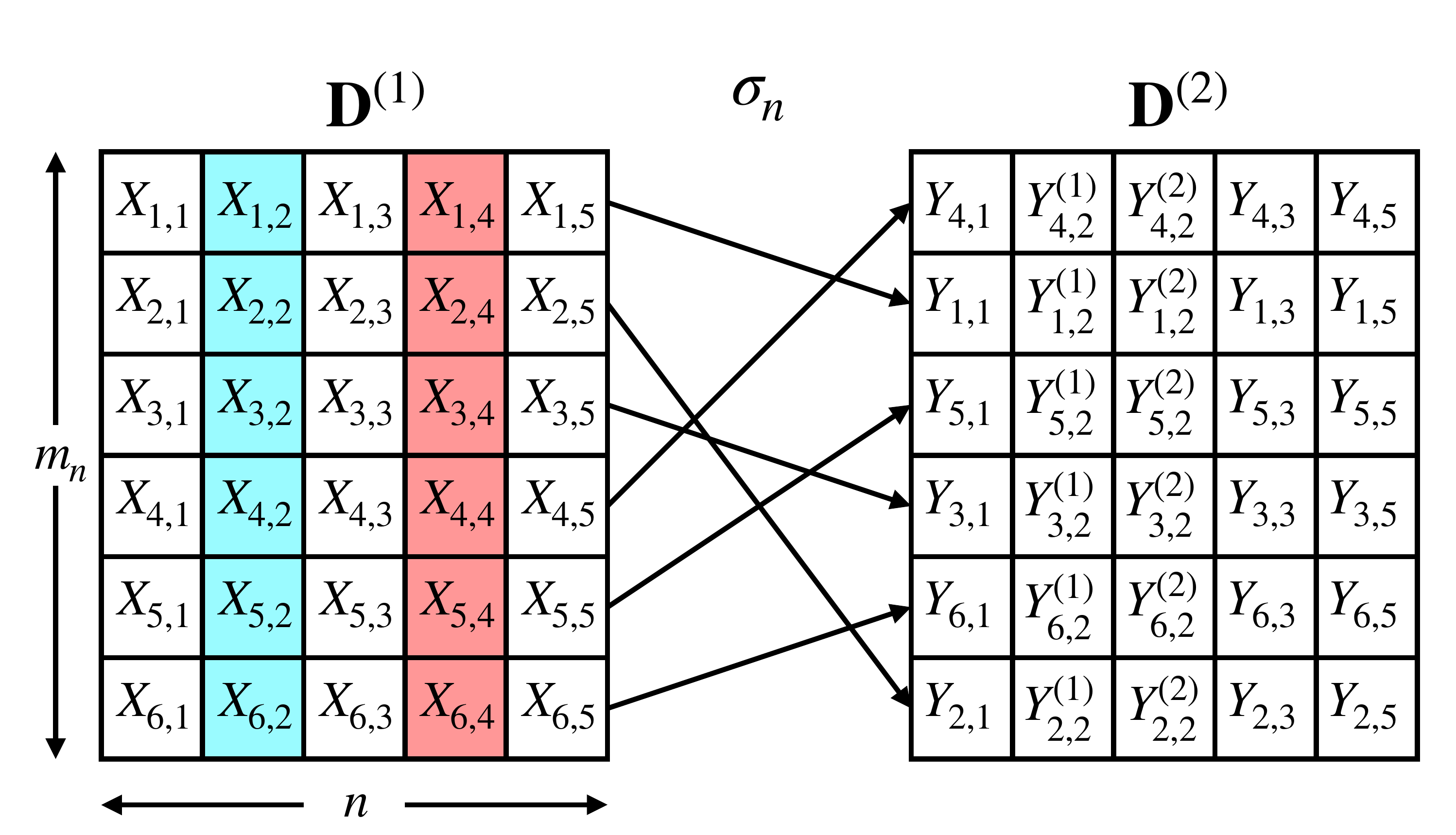}}
\caption{An illustrative example of database matching under column repetitions. The column coloured in red is repeated zero times, \emph{i.e.,} deleted, whereas the column coloured in blue is repeated twice, \emph{i.e.,} replicated. $Y_{i,2}^{(1)}$ and $Y_{i,2}^{(2)}$ denote noisy copies/replicas of $X_{i,2}$. Our goal is to estimate the correct row permutation ${\sigma_n=\left(\begin{smallmatrix} 
1 & 2 & 3 & 4 & 5 & 6\\
2 & 6 & 4 & 1 & 3 & 5
\end{smallmatrix}\right)}$, by matching the rows of $\mathbf{D}^{(1)}$ and $\mathbf{D}^{(2)}$ without any prior information on the underlying database ($p_{X}$), obfuscation ($p_{Y|X}$), and repetition ($p_S$) distributions.}
\label{fig:intro}
\end{figure}
 
The structure of the rest of this paper is as follows: Section~\ref{sec:problemformulation} introduces the formal statement of the problem. Section~\ref{sec:mainresult} contains our proposed algorithms, states our main result, and contains its proof. Finally, Section~\ref{sec:conclusion} consists of the concluding remarks.

\noindent{\em Notation:} We denote a matrix $\mathbf{D}$ with bold capital letters, and its $(i,j)$\textsuperscript{th} element with $D_{i,j}$. A set is denoted by a calligraphic letter, \emph{e.g.}, $\mathfrak{X}$. $[n]$ denotes the set of integers $\{1,\dots,n\}$. Asymptotic order relations are used as defined in~\cite[Chapter 3]{cormen2022introduction}. All logarithms are base 2. $H(.)$ and $I(.;.)$ denote the Shannon entropy and the mutual information~\cite[Chapter 2]{cover2006elements}, respectively.

\section{Problem Formulation}
\label{sec:problemformulation}

We use the following definitions, most of which are borrowed from~\cite{bakirtas2022seeded} to formalize our problem.

\begin{defn}{\textbf{(Anonymized Database)}}\label{defn:unlabeleddb}
An ${(m_n,n,p_{X})}$ \emph{anonymized database} ${\mathbf{D}=\{X_{i,j}\in\mathfrak{X}\}, (i,j)\in[m_n]\times [n]\}}$ is a randomly generated ${m_n\times n}$ matrix with $\smash{X_{i,j}\overset{\text{i.i.d.}}{\sim} p_X}$, where $p_X$ has a finite discrete support $\mathfrak{X}=\{1,\dots,|\mathfrak{X}|\}$.
\end{defn}

\begin{defn}{\textbf{(Column Repetition Pattern)}}\label{defn:repetitionpattern}
The \emph{column repetition pattern} $S^n=\{S_1,S_2,...,S_n\}$ is a random vector with $\smash{S_i\overset{\text{i.i.d.}}{\sim} p_S}$, where $p_S$ has a finite integer support ${\{0,\dots,s_{\max}\}}$.
\end{defn}

\begin{defn}{\textbf{(Anonymization Function)}}
    The \emph{anonymization function} $\sigma_n$ is a uniformly-drawn permutation of $[m_n]$.
\end{defn}

\begin{defn}{\textbf{(Labeled Correlated Database)}}\label{defn:labeleddb}
Let $\mathbf{D}^{(1)}$, $S^n$ and $\sigma_n$ be a mutually-independent ${(m_n,n,p_{X})}$ anonymized database, repetition pattern and anonymization function triplet. Let $p_{Y|X}$ be a conditional probability distribution with both $X$ and $Y$ taking values from $\mathfrak{X}$. Given $\smash{\mathbf{D}^{(1)}}$, ${S}^n$, $\sigma_n$ and $p_{Y|X}$, $\smash{\mathbf{D}^{(2)}}$ is called the \emph{labeled correlated database} if the respective $(i,j)$\textsuperscript{th} entries $X_{i,j}$ and $Y_{i,j}$ of $\mathbf{D}^{(1)}$ and $\mathbf{D}^{(2)}$ have the following relation:
\begin{align}
    Y_{\sigma_n(i),j}&=
    \begin{cases}
      E , &  \text{if } S_{j}=0\\
      Z^{S_j} & \text{if } S_{j}\ge 1
    \end{cases} \quad \forall i\in[m_n],\:\forall j\in[n]
\end{align}
where $Z^{S_j}$ is a row vector consisting of $S_j$ noisy replicas of $X_{i,j}$ with the following conditional probability distribution
\vspace{-1em}
\begin{align}
    \Pr\left(Z^{S_j}=z^{S_j}|X_{i,j}=x\right)
    &=\prod\limits_{l=1}^{S_j} p_{Y|X}\left(z_l |x\right)\label{eq:noiseiid}
\end{align}
where $z^{S_j}=z_1,\dots,z_{S_j}$ and ${ Y_{\sigma_n(i),j}=E}$ corresponds to $ Y_{\sigma_n(i),j}$ being the empty string.

Note that $S_j$ indicates the times the $j$\textsuperscript{th} column of $\mathbf{D}^{(1)}$ is repeated. When $S_j=0$, the $j$\textsuperscript{th} column of $\mathbf{D}^{(1)}$ is said to be \emph{deleted} and when $S_j>1$, the $j$\textsuperscript{th} column of $\mathbf{D}^{(1)}$ is said to be \emph{replicated}.

The $i$\textsuperscript{th} row $X_i$ of $\mathbf{D}^{(1)}$ and the $\sigma_n(i)$\textsuperscript{th} row $Y_{\sigma_n(i)}$ of $\mathbf{D}^{(2)}$ are called \emph{matching rows}.
\end{defn}

\begin{figure}[t]
\centerline{\includegraphics[width=0.5\textwidth,trim={0cm 25cm 0cm 0cm},clip]{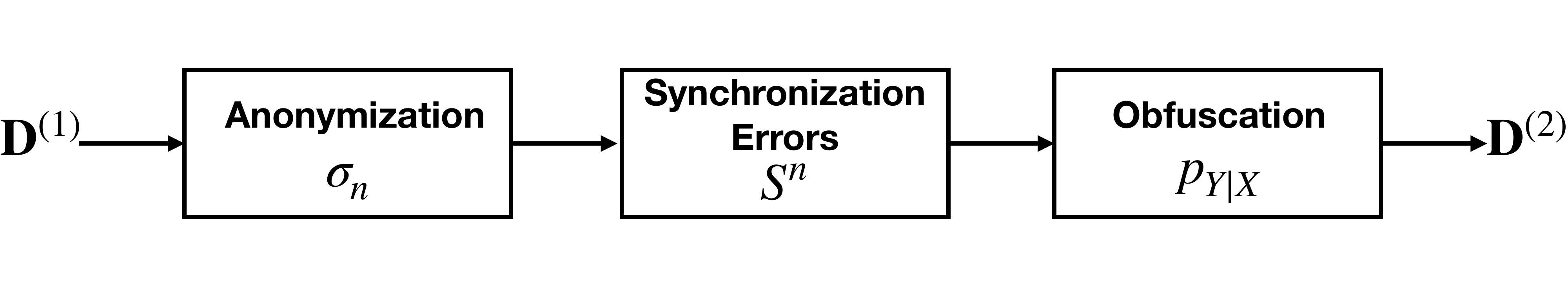}}
\caption{Relation between the unlabeled database $\mathbf{D}^{(1)}$ and the labeled noisy repeated one, $\mathbf{D}^{(2)}$.}
\label{fig:systemmodel}
\end{figure}

The relationship between $\mathbf{D}^{(1)}$ and $\mathbf{D}^{(2)}$, as described in Definition~\ref{defn:labeleddb}, is illustrated in Figure~\ref{fig:systemmodel}.

\begin{rem}{\textbf{(Assumptions)}}
\begin{enumerate}[label=\textbf{(\alph*)}]
    \item The fact that $\smash{X_{i,j}}$ and $\smash{S_i}$ are \emph{i.i.d.} can be checked through the Markov order estimation algorithm of~\cite{morvai2005order} with a probability of error vanishing in $n$. Thus from now on, we assume that the \emph{i.i.d.} nature of $\smash{X_{i,j}}$ and $\smash{S_i}$ is known, while the distributions $p_X$ and $p_S$ are not.
    \item Since $|\mathfrak{X}|$ and $s_{\max}$ do not depend on $n$, they can easily be estimated with a probability of error vanishing in $n$. Therefore, we will assume that $|\mathfrak{X}|$ and $s_{\max}$ are known.
    \item In this work, we assume a memoryless noise model, so that the conditional independence of the noisy replicas stated in \eqref{eq:noiseiid} is known, whereas the noise distribution $p_{Y|X}$ is not.
\end{enumerate}
\end{rem}

As often done in both the graph matching~\cite{shirani2017seeded} and the database matching~\cite{bakirtas2022seeded} literatures, we will assume the availability of a set of already-matched row pairs called \emph{seeds}, to be used in the detection of the underlying repetition pattern.

\begin{defn}{\textbf{(Seeds)}}
    Given a pair of anonymized and labeled correlated databases $(\mathbf{D}^{(1)},\mathbf{D}^{(2)})$, a \emph{seed} is a correctly-matched row pair with the same underlying repetition pattern. A \emph{batch of $\Lambda_n$ seeds} is a pair of seed matrices of respective sizes $\Lambda_n\times n$ and $\Lambda_n\times \sum_{j=1}^n S_j$.
\end{defn}
For the sake of notational brevity, we assume that the seed matrices $\mathbf{G}^{(1)}$ and $\mathbf{G}^{(2)}$ are not submatrices of $\mathbf{D}^{(1)}$ and $\mathbf{D}^{(2)}$. Throughout this work, we will assume a seed size $\Lambda_n=\omega(\log n)=\omega(\log\log m_n)$ which is double-logarithmic with the number of users $m_n$.

As done in~\cite{shirani8849392,bakirtas2021database,bakirtas2022database,bakirtas2022matching,bakirtas2022seeded}, we utilize the database growth rate, defined below, as the main performance metric.
\begin{defn}\label{defn:dbgrowthrate}{\textbf{(Database Growth Rate)}}
The \emph{database growth rate} $R$ of an $(m_n,n,p_X)$ anonymized database is defined as 
\begin{align}
    R&=\lim\limits_{n\to\infty} \frac{1}{n}\log m_n
\end{align}
\end{defn}

Similar to~\cite{shirani8849392,bakirtas2021database,bakirtas2022seeded,bakirtas2022matching}, our goal is to characterize the supremum of the achievable database growth rates allowing the \emph{almost-perfect} recovery of the anonymization function $\sigma_n$. However, unlike~\cite{shirani8849392,bakirtas2021database,bakirtas2022seeded,bakirtas2022matching}, we consider the case \emph{when the underlying distributions $p_X$, $p_{Y|X}$ and $p_S$ are not provided a priori}. More formally, \say{almost-perfect recovery} corresponds to the construction of the estimate $\hat{\sigma}_n$ such that
\begin{align}
\lim\limits_{n\to\infty}\Pr\left(\boldsymbol{\sigma}_n(J)\neq\hat{\sigma}_n(J)\right)&\to 0
\end{align}
where $J\sim\text{Unif}([m_n])$.

\section{De-Anonymization Algorithm and Achievability}
\label{sec:mainresult}
In this section, we present our main result in Theorem~\ref{thm:mainresult} on the achievable database growth rates when no prior information is provided on $p_{X}$, $p_{Y|X}$, and $p_S$.
\begin{thm}{\textbf{(Main Result)}}\label{thm:mainresult}
    Consider an anonymized and labeled correlated database pair, with underlying database distributions $p_{X,Y}$ and a column repetition distribution $p_S$ which are assumed to be not known a priori. Given a seed size $\Lambda_n=\omega(\log n)$, any database growth rate $R$ satisfying
    \begin{align}
        R&<I(X;Y^S|S) \label{eq:mainresult}
    \end{align}
    is achievable where $S\sim p_S$, $X~\sim p_X$ and ${Y^S=Y_1,\dots,Y_S}$ with $Y_i|X\overset{\text{i.i.d.}}{\sim} p_{Y|X}$.
\end{thm}

In order to demonstrate the tightness of the achievability result stated in Theorem~\ref{thm:mainresult}, we now compare it to the distribution-aware results derived in~\cite[Theorem~1]{bakirtas2022seeded}.

\begin{thm}{\textbf{(Converse of~\cite[Theorem~1]{bakirtas2022seeded})}}\label{rem:converse}
    Consider an anonymized and labeled correlated database pair, with underlying joint database distributions $p_{X,Y}$ and a column repetition distribution $p_S$. Then, a necessary condition for the existence of a successful de-anonymization scheme is:
    \begin{align}
        R&\le I(X;Y^S|S)
    \end{align}
\end{thm}

Theorems~\ref{thm:mainresult} and~\ref{rem:converse} imply that given a seed size ${\Lambda_n=\omega(\log n)=\omega(\log\log m_n)}$ we can perform matching as if we knew the underlying distributions $p_{X,Y}$ and $p_S$, and the actual column repetition pattern $S^n$ a priori. Hence in the asymptotic regime, not knowing the distributions causes no loss in the matching capacity.

The rest of this section is on the proof of Theorem~\ref{thm:mainresult}. In Section~\ref{subsec:replicadetection}, we present our detection of noisy replicas algorithm and prove its asymptotic performance. Then in Section~\ref{subsec:deletiondetection}, we propose a seeded deletion algorithm and derive a sufficient seed size that guarantees its asymptotic performance. Finally in Section~\ref{subsec:achievability}, we present our de-anonymization algorithm.

\subsection{Noisy Replica Detection}
\label{subsec:replicadetection}
Similar to~\cite{bakirtas2022seeded}, we use the running Hamming distances between the consecutive columns $\smash{C_j^{(2)}}$ and $\smash{C_{j+1}^{(2)}}$ of $\mathbf{D}^{(2)}$, denoted by $H_j$, $j\in[K_n-1]$, where $K_n\triangleq\sum_{j=1}^n S_j$ as a permutation-invariant future of the labeled correlated database. More formally,
\begin{align}
    H_j & \triangleq \sum\limits_{t=1}^{m_n} \mathbbm{1}_{[D^{(2)}_{t,j+1}\neq D^{(2)}_{t,j}]},\hspace{2em} \forall j\in[K_n-1]\label{eq:RHD}
    \end{align}
We first note that 
\begin{align}
H_j&\sim \begin{cases}
        \text{Binom}(m_n,p_0),& \text{if }C^{(2)}_j \indep C^{(2)}_{j+1} \\
        \text{Binom}(m_n,p_1), & \text{otherwise}
    \end{cases}
\end{align}
where 
\begin{align}
    p_0 &\triangleq 1-\sum\limits_{y\in\mathfrak{X}} p_Y(y)^2\\
    p_1 &\triangleq 1-\sum\limits_{x\in\mathfrak{X}} p_X(x) \sum\limits_{y\in\mathfrak{X}} p_{Y|X}(y|x)^2
\end{align}
From~\cite[Lemma 1]{bakirtas2022database}, we know that as long as the databases are correlated, \emph{i.e.,} $p_{X,Y}\neq p_X p_Y$, we have $p_0>p_1$ for any $p_{X,Y}$. Thus, as long as $p_{X,Y}\neq p_X p_Y$, replicas can be detected based on the Hamming distances $H_j$ similar to~\cite{bakirtas2022seeded,bakirtas2022database}. However, the algorithm in~\cite{bakirtas2022seeded} depends on the choice of a threshold that depends on $p_{X,Y}$ through $p_0$ and $p_1$. In Algorithm~\ref{alg:noisyreplicadetection}, we propose the following modification: We first construct the estimates $\hat{p}_0$ and $\hat{p}_1$ for the respective parameters $p_0$ and $p_1$ through the moment estimator proposed by Blischke in~\cite{blischke1962moment}. Note that we can use this estimator because the Binomial mixture is guaranteed to have two distinct components. More formally, the distribution of $H_j$ conditioned on $S^n$ is given by
    \begin{align}
    \Pr(H_j=h|S^n) &= \binom{m_n}{h} [\alpha p_0^{h} (1-p_0)^{m_n-h}\notag\\
    &\hspace{5em}+(1-\alpha) p_1^h (1-p_1)^{m_n-h}]
\end{align}
for $h=0,\dots,m_n$ where the mixing parameter $\alpha$ is given by
\begin{align}
    \alpha &= \frac{1}{K_n-1}\left(n-\sum\limits_{j=1}^n \mathbbm{1}_{[S_j=0]}\right)
\end{align}
Since $p_S$, and in turn $\delta$ and $\alpha$, are constant in $n$, it can easily be verified that as $n\to\infty$, $\alpha\overset{p}{\to} \frac{1-\delta}{\mathbb{E}[S]}$. Hence it is bounded away from both 0 and 1, suggesting that the moment estimator of~\cite{blischke1962moment} and in turn Algorithm~\ref{alg:noisyreplicadetection} can be used to detect the replicas. More formally, for any $\epsilon>0$.
\begin{align}
   \lim\limits_{n\to\infty} \Pr\left(\left|\alpha-\frac{1-\delta}{\mathbb{E}[S]}\right|>\epsilon\right)&=0.
\end{align}

\begin{algorithm}[t]
\caption{Noisy Replica Detection Algorithm}\label{alg:noisyreplicadetection}
\Input{$(\mathbf{D},m_n,K_n)$}
\Output{isReplica}
$H\gets $ RunningHammingDist($\mathbf{D}$)\Comment*[r]{Eq.~\eqref{eq:RHD}}
$(\hat{p}_0,\hat{p}_1) \gets $EstimateParams$(H)$\Comment*[r]{See~\cite{blischke1962moment}}
$\tau \gets \frac{\hat{p}_0+\hat{p}_1}{2}$\Comment*[r]{Threshold}
isReplica $\gets \varnothing$\;

\For{$j = 1$ \KwTo $K_n-1$}{
  \eIf{$H[j]\le m_n \tau$}{
    isReplica$[j] \gets$ TRUE\;
  }{
      isReplica$[j] \gets $ FALSE\;
    }
  
}
\end{algorithm}

The following lemma states that this algorithm has a vanishing error probability.
\begin{lem}{\textbf{(Noisy Replica Detection)}}\label{lem:replicadetection}
    Algorithm~\ref{alg:noisyreplicadetection} has a vanishing probability of replica detection error, as long as $m_n=\omega(\log n)$.
\end{lem}

\begin{proof}
The estimator proposed in~\cite{blischke1962moment} works as follows: Define the $k$\textsuperscript{th} sample factorial moment $F_k$ as
\begin{align}
    F_k&\triangleq \frac{1}{K_n-1} \sum\limits_{j=1}^{K_n-1} \prod\limits_{i=0}^{k-1} \frac{H_j-i}{m_n-i},\hspace{1em} \forall k\in[m_n]\label{eq:factorialmoments}
\end{align}
and let 
\begin{align}
    A &\triangleq \frac{F_3-F_1 F_2}{F_2-F_1^2}
\end{align}
Then the respective estimators $\hat{p}_0$ and $\hat{p}_1$ for $p_0$ and $p_1$ can be constructed as:
\begin{align}
    \hat{p}_0 &= \frac{A+\sqrt{A^2-4 A F_1 + 4 F_2}}{2}\\
    \hat{p}_1 &= \frac{A-\sqrt{A^2-4 A F_1 + 4 F_2}}{2}\label{eq:paramest}
\end{align}
From~\cite{blischke1962moment}, we get $\hat{p}_i\overset{p}{\to}p_i$ and in turn $\tau\overset{p}{\to}\frac{p_0+p_1}{2}$. Thus for large $n$, $\tau$ is bounded away from $p_0$ and $p_1$. At this stage, we are ready to finish the proof following the same steps taken in the proof of~\cite[Lemma~1]{bakirtas2023database}, which we provide below for the sake of completeness.

Let $A_j$ denote the event that $\smash{C^{(2)}_{j}}$ and $\smash{C^{(2)}_{j+1}}$ are noisy replicas and $B_j$ denote the event that the algorithm infers $\smash{C^{(2)}_{j}}$ and $\smash{C^{(2)}_{j+1}}$ as replicas. Via the union bound, we can upper bound the total probability of replica detection error as
\begin{align}
    \Pr(\bigcup\limits_{j=1}^{{K_n}-1} E_j)&\le \sum\limits_{j=1}^{{K_n}-1} \Pr(A_j ^c) \Pr(B_j|A_j ^c)+ \Pr(A_j)  \Pr(B_j^c|A_j)\label{eq:replicadetectionbound}
\end{align}
where $E_j$ denotes the replica detection event for $\smash{C^{(2)}_j}$ and $\smash{C^{(2)}_{j+1}}$.

Observe that conditioned on $A_j^c$, $H_j\sim\text{Binom}\smash{(m_n,p_0)}$ and conditioned on $A_j$, $H_j\sim\text{Binom}(m_n,p_1)$. Then, from the Chernoff bound~\cite[Lemma 4.7.2]{ash2012information}, we get
\begin{align}
    \Pr(B_j|A_j ^c)&\le 2^{-m_n D\left(\tau\|\smash{p_0}\right)}\label{eq:chernoff1}\\
    \Pr(B_j^c|A_j)&\le 2^{-m_n D\left(1-\tau \|1-p_1\right)}\label{eq:chernoff2}
\end{align}

Thus, we obtain
\begin{align}
    \Pr(\bigcup\limits_{j=1}^{{K_n}-1} E_j)&\le ({K_n}-1)\left[ 2^{-m_n D\left(\tau\|\smash{p_0}\right)}+ 2^{-m_n D\left(1-\tau\|1-p_1\right)}\right]\label{eq:replicadetectionlast}
\end{align}

Since the RHS of \eqref{eq:replicadetectionlast} has $2{K_n}=O(n)$ terms decaying exponentially in~$m_n$, for any $m_n=\omega(\log n)$ we have 
\begin{align}
    \Pr(\bigcup\limits_{j=1}^{{K_n}-1} E_j) \to 0 \:\text{ as } n\to\infty.
\end{align}

Observing that $n\sim \log m_n$, we get
\begin{align}
    \lim\limits_{n\to \infty}\Pr(\text{Noisy replica detection error})&= 0.
\end{align}
\end{proof}

Note that the condition in Lemma~\ref{lem:replicadetection} is automatically satisfied since $m_n$ is exponential in $n$ (Definition~\ref{defn:dbgrowthrate}). Finally, we stress that as opposed to deletion detection, discussed in Section~\ref{subsec:deletiondetection}, no seeds are necessary for replica detection.

\subsection{Deletion Detection}
\label{subsec:deletiondetection}
In this section, we propose a deletion detection algorithm that utilizes the seeds. Since the replica detection algorithm of Section~\ref{subsec:replicadetection} (Algorithm~\ref{alg:noisyreplicadetection}) has a vanishing probability of error, for notational simplicity we will focus on a deletion-only setting throughout this subsection. Let $\mathbf{G}^{(1)}$ and $\mathbf{G}^{(2)}$ be the seed matrices with respective sizes $\Lambda_n\times n$ and $\Lambda_n\times \Tilde{K}_n$, and denote the $j$\textsuperscript{th} column of $\mathbf{G}^{(r)}$ with $G_j^{(r)}$, $r=1,2$ where $\Tilde{K}_n\triangleq\sum_{j=1}^n \mathbbm{1}_{[S_j\neq 0]}$. Furthermore, for the sake of brevity, let $L_{i,j}$ denote the Hamming distance between $\smash{G^{(1)}_i}$ and $\smash{G^{(2)}_j}$ for $(i,j)\in[n]\times [\Tilde{K}_n]$. More formally, let 
\begin{align}
    L_{i,j}&\triangleq \sum\limits_{t=1}^{\Lambda_n} \mathbbm{1}_{[G^{(1)}_{t,i}\neq G^{(2)}_{t,j}]}\label{eq:crosshammingdist}
\end{align}
Observe that 
\begin{align}
L_{i,j}&\sim \begin{cases}
        \text{Binom}(\Lambda_n,q_0),& G^{(1)}_i \indep G^{(2)}_{j} \\
        \text{Binom}(\Lambda_n,q_1), & \text{otherwise}
    \end{cases}
\end{align}where
\begin{align}
    q_0 &= 1-\sum\limits_{x\in\mathfrak{X}} p_X(x) p_Y(x)\\
    q_1 &= 1-\sum\limits_{x\in\mathfrak{X}} p_{X,Y}(x,x)
\end{align}

Thus, we have a problem seemingly similar to the one in Section~\ref{subsec:replicadetection}. However, we cannot utilize similar tools because of the following: 
\emph{i)} Recall that the two components of the Binomial mixture discussed in Section~\ref{subsec:replicadetection} were distinct for any underlying joint distribution $p_{X,Y}$ as long as the databases are correlated, \emph{i.e.,} $p_{X,Y}\neq p_X p_Y$. Unfortunately, the same idea does not automatically work here as demonstrated by the following example: Suppose $X_{i,j}\sim\text{Unif}(\mathfrak{X})$, and the transition matrix $\mathbf{P}$ associated with $p_{Y|X}$ has unit trace. Then, 
\begin{align}
    q_0-q_1&=\sum\limits_{x\in\mathfrak{X}} p_{X,Y}(x,x) - p_X(x) p_Y(x)\\
    &= \frac{1}{|\mathfrak{X}|} \sum\limits_{x\in\mathfrak{X}} p_{Y|X}(x|x) - p_Y(x)\\
    &= \frac{1}{|\mathfrak{X}|}(\text{tr}(\mathbf{P})-1)\\
    &= 0
\end{align}
In~\cite{bakirtas2022seeded}, we overcame this problem using the following modification: Based on $p_{X,Y}$, we picked a bijective remapping $\Phi\in \mathfrak{S}(\mathfrak{X})$ and applied it to all the entries of $\mathbf{G}^{(2)}$ before computing the Hamming distances $L_{i,j}$, where $\mathfrak{S}(\mathfrak{X})$ denotes the symmetry group of $\mathfrak{X}$. Denoting the resulting version of the Hamming distance $L_{i,j}$ by $L_{i,j}(\Phi)$, we proved in~\cite[Lemma 2]{bakirtas2022seeded} that there as long as $p_{X,Y}\neq p_X p_Y$, there exists $\Phi\in \mathfrak{S}(\mathfrak{X})$ such that the Binomial mixture distribution associated with $L_{i,j}(\Phi)$ has two distinct components with respective success parameters $q_0(\Phi)$ and $q_1(\Phi)$. In other words, we have
\begin{align}
L_{i,j}(\Phi)&\sim \begin{cases}
        \text{Binom}(m_n,q_0(\Phi)),& G^{(1)}_i \indep G^{(2)}_{j} \\
        \text{Binom}(m_n,q_1(\Phi)), & \text{otherwise}
    \end{cases}
    \label{eq:Ldist}
\end{align}
and $q_0(\Phi)\neq q_1(\Phi)$. We will call such $\Phi$ a \emph{useful remapping}. 

\emph{ii)} In the known-distribution setting, we chose the useful remapping $\Phi$ and threshold $\tau_n$ for Hamming distances based on $p_{X,Y}$. In Section~\ref{subsec:replicadetection}, we solved the distribution-agnostic case via parameter estimation in Binomial mixtures. However, the same approach does not work here. Suppose the $j$\textsuperscript{th} retained column $\smash{G_{j}^{(2)}}$ of $\mathbf{G}^{(2)}$ is correlated with $\smash{G_{r_j}^{(1)}}$. Then the $j$\textsuperscript{th} column of $\mathbf{L}(\Phi)$ will have a Binom($\Lambda_n,q_1(\Phi)$) component in the $r_j$\textsuperscript{th} row, whereas the remaining $n-1$ rows will contain Binom($\Lambda_n,q_0(\Phi)$) components, as described in~\eqref{eq:Ldist} and illustrated in Figure~\ref{fig:deletiondetection}. 
Hence, it can be seen that the mixture parameter $\beta$ of this Binomial mixture distribution approaches 1 since
\begin{align}
    \beta &= \frac{(n-1)\Tilde{K}_n}{n \Tilde{K}_n} = 1-\frac{1}{n}
\end{align}
This imbalance prevents us from performing a parameter estimation as done in Algorithm~\ref{alg:noisyreplicadetection}.
\begin{figure}[t]
\centerline{\includegraphics[width=0.5\textwidth]{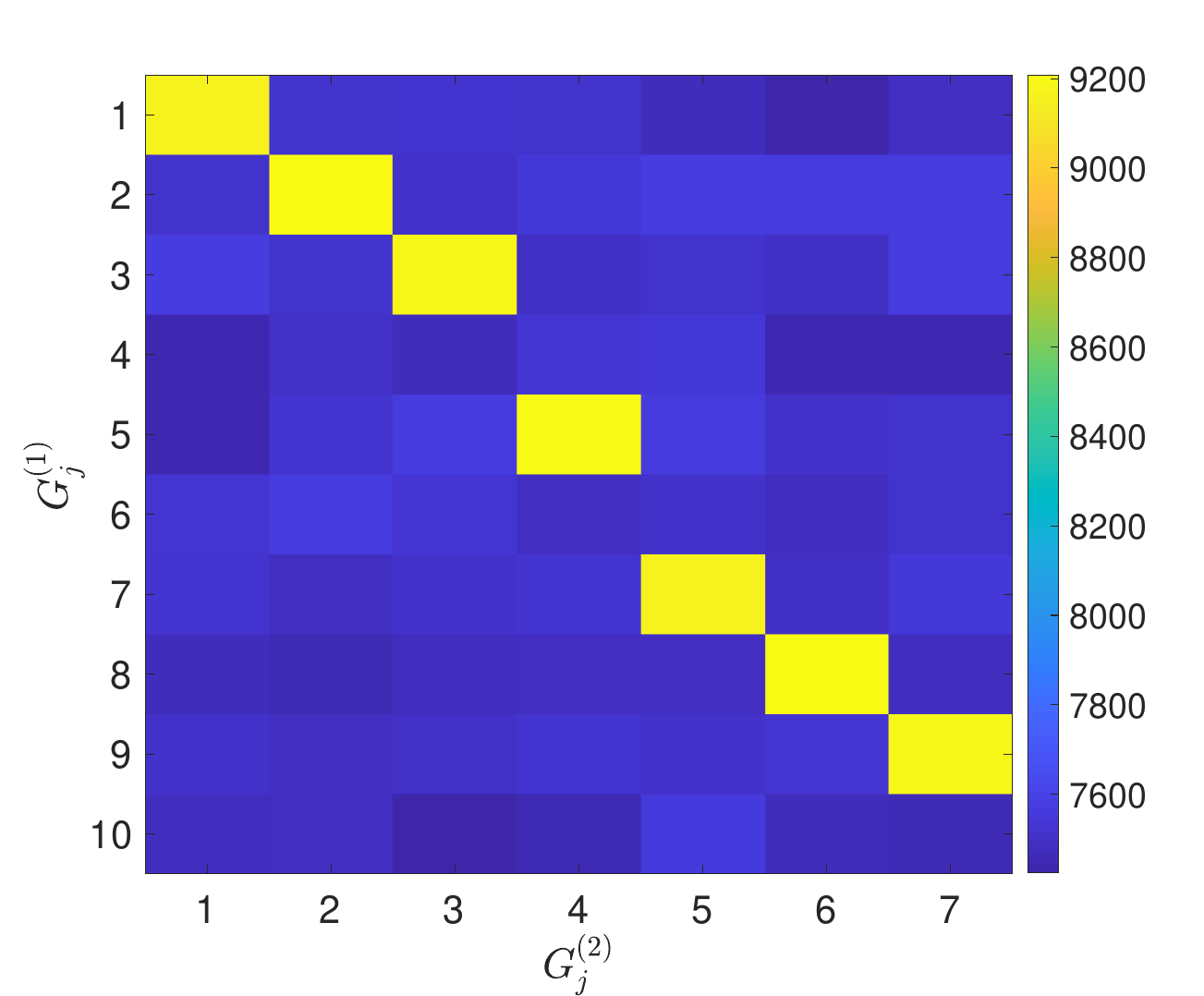}}
\caption{Hamming distances between the columns of $\mathbf{G}^{(1)}$ and $\mathbf{G}^{(2)}$ with $n=10$, $\Tilde{K}_n=7$ and $\Lambda_n=10^4$ for $q_0\approx 0.76$ and $q_1\approx 0.92$. The $(i,j)$\textsuperscript{th} element corresponds to $L_{i,j}$, with the color bar indicating the approximate values. It can be seen that there are no outliers in the $4$\textsuperscript{th}, $6$\textsuperscript{th}, and $10$\textsuperscript{th} rows. Hence, it can be inferred that $I_{\text{del}}=(4,6,10)$.}
\label{fig:deletiondetection}
\end{figure}

We propose to exploit the aforementioned observation that for a useful mapping $\Phi$, in each column of $\mathbf{L}(\Phi)$, there is exactly one element with a different underlying distribution, while the remaining $n-1$ entries are \emph{i.i.d.}, rendering this entry an \emph{outlier}. Note that $L_{i,j}(\Phi)$ being an outlier corresponds to $G_i^{(1)}$ and $G_j^{(2)}$ being correlated, and in turn $S_i\neq 0$. On the other hand, we stress that the lack of outliers in any given column of $\mathbf{L}(\Phi)$ implies that $\Phi$ is useless. Thus, it can easily be seen that Algorithm~\ref{alg:deletiondetection} is capable of deciding whether a given remapping is useful or not. In fact, the algorithm sweeps over all elements of $\mathfrak{S}(\mathfrak{X})$ until we encounter a useful one.

To detect the outliers in $\mathbf{L}(\Phi)$, we propose to use the distances of $L_{i,j}(\Phi)$ to the sample mean $\mu(\Phi)$ of $\mathbf{L}(\Phi)$ where
\begin{align}
    \mu(\Phi) &\triangleq \frac{1}{n \Tilde{K}_n}\sum\limits_{i=1}^{n} \sum\limits_{j=1}^{\Tilde{K}_n} L_{i,j}(\Phi)\label{eq:samplemean}
\end{align}
As given in Algorithm~\ref{alg:deletiondetection}, if these distances are lower than $\hat{\tau}_n$, we detect retention \emph{i.e.,} non-deletion.

Note that this step is equivalent to utilizing
Z-scores (also known as standard scores), a well-studied concept in statistical outlier detection~\cite{moore2007basic}, where the distances to the sample mean are also divided by the sample standard deviation. In Algorithm~\ref{alg:deletiondetection}, for the sake of brevity, we will avoid such division.

\begin{algorithm}[t]
\caption{Seeded Deletion Detection Algorithm}\label{alg:deletiondetection}
\Input{$(\mathbf{G}^{(1)},\mathbf{G}^{(2)},\Lambda_n,n, \Tilde{K}_n,\mathfrak{X})$}
\Output{retentionIdx}
$\mathfrak{S}(\mathfrak{X})\gets$ SymmetryGroup($\mathfrak{X}$)\;

$\hat{\tau}_n\gets 2 \Lambda_n^{2/3} (\log n)^{1/3}$\Comment*[r]{Threshold}

\For{$s \gets 1$ \KwTo $|\mathfrak{X}|!$}{
retentionIdx$\gets\varnothing$\;
$\Phi\gets\mathfrak{S}(\mathfrak{X})[s]$\Comment*[r]{Pick a remapping.} 
$\mathbf{L}(\Phi) \gets $ HammDist($\mathbf{G}^{(1)},\mathbf{G}^{(2)}$)\Comment*[r]{Eq. \eqref{eq:crosshammingdist}}
$\mu(\Phi)\gets$ SampleMean($\mathbf{L}(\Phi) $)\Comment*[r]{Eq.~\eqref{eq:samplemean}}
$\mathbf{M}(\Phi) \gets |\mathbf{L}(\Phi)-\mu(\Phi)|$\;

\For{$j\gets1$ \KwTo $\Tilde{K}_n$}{
count $\gets0$\; 
\For{$i\gets1$ \KwTo $n$}{
\If{$\mathbf{M}(\Phi)[i][j]\le\hat{\tau}_n$}{
retentionIdx$\gets$ retentionIdx $\cup \{i\} $\;
count $\gets$ count $+$ $1$\;
}

}
\blue{/* count = 0: no outliers ($\Phi$ is useless). */}\\
\blue{/* count $>$ 1: misdetection. */}\\
\eIf{$\textup{count}>1$}{\Return ERROR}{\If{$\textup{count}=0$}{Skip to next $\Phi$\;}}
}
\Return $\hat{I}_R$\;
}

\end{algorithm}

The following lemma states that for sufficient seed size, $\Lambda_n=\omega(\log n)= \omega(\log\log m_n)$, Algorithm~\ref{alg:deletiondetection} works correctly with high probability.
\begin{lem}{\textbf{(Deletion Detection)}}
    Let $I_R=\{j\in[n]:\: S_j\neq 0\}$ be the true retention index set and $\hat{I}_R$ be its estimate output by Algorithm~\ref{alg:deletiondetection}. Then for any seed size $\Lambda_n=\omega(\log n)$, we have
    \begin{align}
        \lim\limits_{n\to\infty} \Pr(\hat{I}_R = I_R) &= 1
    \end{align}
\end{lem}

\begin{proof}
    For now, suppose that $\Phi$ is a useful remapping.
    Start by observing that using Chebyshev's inequality~\cite[Theorem 4.2]{wasserman2004all} it is straightforward to prove that for any $\epsilon_n>0$
\begin{align}
    \gamma&\triangleq \Pr(|\mu(\Phi)-\Lambda_n q_0(\Phi)|>\Lambda_n {\epsilon}_n) \le O\left(\frac{1}{K_n n \Lambda_n {\epsilon}_n}\right)\label{eq:alpha}
\end{align}
Let $I_R=\{r_1,\dots,r_{\tilde{K}_n}\}$ and note that $L_{i,j}\sim\text{Binom}(\Lambda_n,q_1(\Phi))$. Thus, from the Chernoff bound~\cite[Lemma 4.7.2]{ash2012information} we get
\begin{align}
    \beta_{r_j,j}&\triangleq\Pr(|L_{r_j,j}(\Phi)-\Lambda_n q_1(\Phi)|\ge {\epsilon}_n \Lambda_n)\\
    &\le 2^{-\Lambda_n D(q_1(\Phi)-{\epsilon}_n\|q_1(\Phi))}
    \notag\\
    &\hspace{3em}+ 2^{-\Lambda_n D(1-q_1(\Phi)-{\epsilon}_n\|1-q_1(\Phi))} \label{eq:beta}
\end{align}
where $D(p\|q)$ denotes the relative entropy~\cite[Chapter 2.3]{cover2006elements} (in bits)  between two Bernoulli distributions with respective parameters $p$ and $q$. 

Now, for notational brevity, let 
\begin{align}
    f(\epsilon) &\triangleq D(q-\epsilon\|q)\\
    g(\epsilon) &\triangleq D(1-q-\epsilon\|1-q)
\end{align}
Then, one can simply verify the following
\begin{align}
    f^\prime(\epsilon) &= \log\frac{q}{1-q}-\log\frac{q-\epsilon}{1-q-\epsilon}\\
    f^{\prime\prime}(\epsilon)&=\frac{1}{\log e} \left[\frac{1}{q-\epsilon} + \frac{1}{1-q+\epsilon}\right]\\
    g^\prime(\epsilon) &= \log\frac{1-q}{q}-\log\frac{1-q-\epsilon}{q+\epsilon}\\
    g^{\prime\prime}(\epsilon)&=\frac{1}{\log e} \left[ \frac{1}{1-q-\epsilon} + \frac{1}{q+\epsilon}\right]
\end{align}
Observing that 
\begin{align}
    f(0)&=f^\prime(0)=0\\
    g(0)&=g^\prime(0)=0
\end{align}
and performing second-order MacLaurin Series expansions on $f$ and $g$, we get for any $\epsilon<1$
\begin{align}
    f(\epsilon) &= c(q) \epsilon^2 + O(\epsilon^3)\\
    g(\epsilon) &= c(q) \epsilon^2  + O(\epsilon^3)
\end{align}
where
\begin{align}
    c(q)&\triangleq\frac{1}{\log e} \left[\frac{1}{q} + \frac{1}{1-q}\right]
\end{align}

Let $\Lambda_n = \Gamma_n \log n$ and $\epsilon_n =\Gamma_n^{-\nicefrac{1}{3}}$ and pick the threshold as $\hat{\tau}_n = 2\Lambda_n \epsilon_n$. Observe that since $\Gamma_n=\omega_n(1)$, we get
\begin{align}
    \hat{\tau}_n &= 2 \Lambda_n \epsilon_n  = o_n(\Lambda_n)\\
    \Lambda_n \epsilon_n^2 &= \Gamma_n^{\nicefrac{1}{3}} \log n = \omega_n(\log n)
\end{align}
Then, we have
\begin{align}
    \beta_{r_j,j} &\le 2^{1-\Lambda_n (c(q_1(\Phi)) \epsilon_{n}^2+O(\epsilon_n^3))}\\
    &= 2^{1- c(q_1(\Phi))  \Gamma_n^{\nicefrac{1}{3}} \log n+O(\epsilon_n^3))}
\end{align}
Note that with probability at least $1-\gamma-\beta_{r_j,j}$ we have
\begin{align}
    |\mu(\Phi)-\Lambda_n q_0(\Phi)|&\le \Lambda_n\epsilon_n\\
    |L_{r_j,j}(\Phi)-\Lambda_n q_1(\Phi)|&\ge \Lambda_n\epsilon_n
\end{align}
From the triangle inequality, we have
\begin{align}
    |L_{r_j,j}(\Phi)-\mu(\Phi)| &\ge \Lambda_n (|q_1(\Phi)-q_0(\Phi)|-2\epsilon_n)\\
    &\ge \hat{\tau}_n
\end{align}
for large $n$.
Therefore, from the union bound we have
\begin{align}
    \Pr(\exists j\in[\tilde{K}_n]:&|L_{r_j,j}-\mu(\Phi)|\le \hat{\tau}_n)\\
    &\le \gamma + \sum\limits_{j=1}^{\tilde{K}_n} \beta_{r_j,j}\\
    &= \gamma + 2^{\log \tilde{K}_n - 1- c(q_1(\Phi)) \Gamma_n^{\nicefrac{1}{3}} \log n +O(\epsilon_n^3))}
\end{align}
Since $\tilde{K}_n\le n$ and $\Lambda_n=\omega_n(\log n)$, we have
\begin{align}
   \lim\limits_{n\to\infty} \log \tilde{K}_n -c(q_1(\Phi)) \Gamma_n^{\nicefrac{1}{3}} \log n
   &= -\infty
\end{align}
Thus we have
\begin{align}
    \lim\limits_{n\to\infty} \Pr(\exists j\in[\tilde{K}_n]:M_{r_j,j}\le \hat{\tau}_n) &= 0
\end{align}

Next, we look at $i\neq r_j$. Repeating the same steps above, we get
\begin{align}
    \beta_{i,j}&\triangleq\Pr(|L_{i,j}(\Phi)-\Lambda_n q_0(\Phi)|\ge \epsilon_n \Lambda_n)\\
    &\le 2^{-\Lambda_n D(q_0(\Phi)-\epsilon_n\|q_0(\Phi))}
    + 2^{-\Lambda_n D(1-q_0(\Phi)-\epsilon_n\|1-q_0\Phi))}\\
    &= 2^{1- c(q_0(\Phi)) \Gamma_n^{\nicefrac{1}{3}} \log n+O(\epsilon_n^3))}
\end{align}
Again, from the triangle inequality, we get
\begin{align}
    |L_{i,j}(\Phi)-\mu(\Phi)|&\le 2\epsilon_n= \hat{\tau}_n
\end{align}
From the union bound, we obtain
\begin{align}
    \Pr(\exists j&\in[\Tilde{K}_n]\:\exists i\in[n]\setminus\{r_j\}: |L_{i,j}(\Phi)-\mu(\Phi)|\ge \hat{\tau}_n) \\
    &\le \gamma + \sum\limits_{j=1}^{\tilde{K}_n}\sum\limits_{i\neq r_j} \beta_{i,j}\\
    &\le \gamma + n^2 2^{1- c(q_0(\Phi))  \Gamma_n^{\nicefrac{1}{3}} \log n+O(\epsilon_n^3))}
\end{align}
Since $\Lambda_n=\omega(\log n)$, we have 
\begin{align}
   \lim\limits_{n\to\infty} \Pr(\exists j&\in[\Tilde{K}_n]\:\exists i\in[n]\setminus\{r_j\}: |L_{i,j}(\Phi)-\mu(\Phi)|\ge \hat{\tau}_n)=0
\end{align}
Thus, for any useful remapping $\Phi$, the misdetection probability decays to zero as $n\to \infty$.

For any \emph{useless} remapping $\Phi$, following the same steps, one can prove that
\begin{align}
    \Pr(&\text{Useless} \text{ remapping }\Phi \text{ is inferred as useful.})\\
    &\le \gamma + \sum\limits_{i=1}^{n}\sum\limits_{j=1}^{\Tilde{K}_n} \Pr(M_{i,j}\ge \epsilon_n \Lambda_n)\\
    &\le \gamma + n^2 2^{1- c(q_0(\Phi))  \Gamma_n^{\nicefrac{1}{3}} \log n+O(\epsilon_n^3))}\\
    &= o_n(1)
\end{align}
Observing $|\mathfrak{S}(\mathfrak{X})|=|\mathfrak{X}|!=O_n(1)$ concludes the proof.
\end{proof}

\subsection{De-Anonymization Scheme}
\label{subsec:achievability}
In this section, we propose a de-anonymization scheme by combining the detection algorithms Algorithm~\ref{alg:noisyreplicadetection} and Algorithm~\ref{alg:deletiondetection}, and performing a modified version of the typicality-based scheme proposed in~\cite{bakirtas2022seeded}. Then using this scheme we prove the achievability of Theorem~\ref{thm:mainresult}. 

Given the database pair $(\mathbf{D}^{(1)},\mathbf{D}^{(2)})$ and the corresponding seed matrices $(\mathbf{G}^{(1)},\mathbf{G}^{(2)})$, the de-anonymization scheme we propose is as follows: 
\begin{enumerate}[label=\textbf{ \arabic*)},leftmargin=1.3\parindent]
    \item Detect the replicas through Algorithm~\ref{alg:noisyreplicadetection}.
    \item Remove all the extra replica columns from the seed matrix $\mathbf{G}^{(2)}$ to obtain $\tilde{\mathbf{G}}^{(2)}$ and perform seeded deletion detection via Algorithm~\ref{alg:deletiondetection} using $\mathbf{G}^{(1)},\tilde{\mathbf{G}}^{(2)}$. At this step, we have an estimate $\hat{S}^n$ of the column repetition pattern $S^n$.
    \item Based on $\hat{S}^n$ and the matching entries in $\mathbf{G}^{(1)},\tilde{\mathbf{G}}^{(2)}$, obtain an estimate $\hat{p}_{X,Y^S|S}$ of $p_{X,Y^S|S}$
    where
    \begin{align}
     \hat{p}_X(x)&\triangleq \frac{1}{\Lambda_n n} \sum\limits_{i=1}^{\Lambda_n}\sum\limits_{j=1}^{n}\mathbbm{1}_{[G_{i,j}^{(1)}=x]},\hspace{2em}\forall x\in\mathfrak{X}\\
     \hat{p}_{Y|X}(y|x) &=\frac{\sum\limits_{i=1}^{\Lambda_n}\sum\limits_{j=1}^{\tilde{K}_n}\mathbbm{1}_{[G^{(1)}_{i,r_j}=x,\tilde{G}^{(2)}_{i,j}=y ]}}{\sum\limits_{i=1}^{\Lambda_n}\sum\limits_{j=1}^{\tilde{K}_n}\mathbbm{1}_{[\tilde{G}^{(2)}_{i,j}=y ]}} ,\hspace{1em}\forall (x,y)\in\mathfrak{X}^2\\
     \hat{p}_S(s) &=\frac{1}{n}\sum\limits_{j=1}^n \mathbbm{1}_{[S_j=s]},\hspace{5em}\forall s\ge0
    \end{align}
    and construct
    \begin{align}
        \hat{p}_{X,{Y}^S|S}(x,y^s|s)&=\begin{cases}
      \hat{p}_X(x) \mathbbm{1}_{[y^s = \ast]} &\text{if } s=0\\
      \hat{p}_X(x) \prod\limits_{j=1}^s \hat{p}_{Y|X}(y_j|x) &\text{if } s\ge 1
    \end{cases}
    \end{align}
with $y^s=y_1\dots y_s$.
    \item Using $\hat{S}^n$, place markers between the noisy replica runs of different columns to obtain $\tilde{\mathbf{D}}^{(2)}$. If a run has length 0, \emph{i.e.} deleted, introduce a column consisting of erasure symbol $\ast\notin\mathfrak{X}$.
    \item Fix $\epsilon>0$. Match the $l$\textsuperscript{th} row $Y^K_{l}$ of $\tilde{\mathbf{D}}^{(2)}$ with the $i$\textsuperscript{th} row $X^n_i$ of {${\mathbf{D}}^{(1)}$}, if $X_i$ is the only row of {${\mathbf{D}}^{(1)}$} jointly $\epsilon$-typical~\cite[Chapter 7.6]{cover2006elements} with $Y^K_l$ according to $\hat{p}_{X,Y^S,S}$, assigning $\hat\sigma_n(i)=l$.
Otherwise, declare an error.
\end{enumerate}
Let $\kappa_n^{(1)}$ and $\kappa_n^{(2)}$ be the error probabilities of the noisy replica detection (Algorithm~\ref{alg:noisyreplicadetection}) and the seeded deletion (Algorithm~\ref{alg:deletiondetection}) algorithms, respectively.
By the Law of Large Numbers, we have 
\begin{align}
    \hat{p}_{X,Y^S|S}\overset{p}{\to} p_{X,Y^S|S}
\end{align}
and by the Continuous Mapping Theorem~\cite[Theorem 2.3]{van2000asymptotic} we have
\begin{align}
    \hat{H}(X,Y^S|S)&\overset{p}{\to} H(X,Y^S|S)\\
I(\hat{X};\hat{Y}^{\hat{S}}|\hat{S})& \overset{p}{\to} I(X;Y^S|S)
\end{align}
where $\hat{H}(X,Y^S|S)$ and $\hat{I}(X,Y^S|S)$ denote the conditional joint entropy and conditional mutual information associated with $\hat{p}_{X,Y^S|S}$, respectively. Thus, for any $\epsilon>0$ we have
\begin{align}
\kappa_n^{(3)}&\triangleq\Pr(|\hat{H}(X,Y^S|S)-H(X,Y^S|S)|>\epsilon)\overset{n\to\infty}{\longrightarrow}0\\
\kappa_n^{(4)}&\triangleq\Pr(|\hat{I}(X,Y^S|S)-I(X,Y^S|S)|>\epsilon)\overset{n\to\infty}{\longrightarrow}0
\end{align}

Using a series of union bounds and triangle inequalities, the probability of error of the de-anonymization scheme can be bounded as 
\begin{align}
    \Pr(\text{error})&\le 2^{-n(I(X;Y^S,S)-4 \epsilon-R)}+\epsilon + \sum\limits_{i=1}^4 \kappa_n^{(i)}\\
    &\le \epsilon
\end{align}
as $n\to\infty$ as long as $R<I(X;Y^S,S)-4\epsilon$, concluding the proof of the main result.

 \section{Conclusion}
\label{sec:conclusion}
In this work, we have investigated the distribution-agnostic database de-anonymization problem under synchronization errors and noise. We have showed that the noisy replica detection algorithm of~\cite{bakirtas2022seeded} tailored for specific $p_{X,Y}$ could be adjusted to work in tandem with a moment estimator to accommodate the unknown $p_{X,Y}$. We have proposed an outlier-detection-based seeded deletion detection algorithm and showed that a seed size growing double logarithmic with the number of rows is sufficient for the correct estimation of the deletion pattern. Finally, we have used a joint-typicality-based de-anonymization scheme utilizing the estimated distributions. Overall, our results show that the resulting achievable database growth rate is equal to the matching capacity derived when full information on the underlying distributions is available.

\bibliography{references}
\bibliographystyle{IEEEtran}

\end{document}